\documentclass[sigconf]{acmart}

\usepackage{amsmath}
\usepackage{amsthm}
\usepackage{enumitem}

\theoremstyle{plain}
\newtheorem{theorem}{Theorem}

\newtheorem{corollary}{Corollary}

\theoremstyle{definition}
\newtheorem{definition}{Definition}

\setcopyright{acmlicensed}
\copyrightyear{2026}
\acmYear{2026}
\acmDOI{XXXXXXX.XXXXXXX}

\acmConference[Conference acronym 'XX]{Make sure to enter the correct
  conference title from your rights confirmation email}{June 03--05,
  2018}{Woodstock, NY}
\acmISBN{978-1-4503-XXXX-X/2018/06}



\title{On the Equivalence Between Auto-Regressive Next Token Prediction and Full-Item-Vocabulary Maximum Likelihood Estimation in Generative Recommendation--A Short Note}
\author{Yusheng Huang}
\affiliation{%
  \institution{Kuaishou Technology}
  \city{Beijing}
  \country{China}}
\email{huangyusheng@kuaishou.com}
\authornote{Corresponding author}

\author{Shuang Yang}
\affiliation{%
  \institution{Kuaishou Technology}
  \city{Beijing}
  \country{China}}
\email{yangshuang08@kuaishou.com}

\author{Zhaojie Liu}
\affiliation{%
  \institution{Kuaishou Technology}
  \city{Beijing}
  \country{China}}
\email{zhaotianxing@kuaishou.com}

\author{Han Li}
\affiliation{%
  \institution{Kuaishou Technology}
  \city{Beijing}
  \country{China}}
\email{lihan08@kuaishou.com}

\begin{document}

\begin{abstract}
Generative recommendation (GR) has emerged as a widely adopted paradigm in industrial sequential recommendation. Current GR systems follow a similar pipeline: tokenization for item indexing, next-token prediction as the training objective and auto-regressive decoding for next-item generation. However, existing GR research mainly focuses on architecture design and empirical performance optimization, with few rigorous theoretical explanations for the working mechanism of auto-regressive next-token prediction in recommendation scenarios. 

In this work, we formally prove that \textbf{the k-token auto-regressive next-token prediction (AR-NTP) paradigm is strictly mathematically equivalent to full-item-vocabulary maximum likelihood estimation (FV-MLE)}, under the core premise of a bijective mapping between items and their corresponding k-token sequences. We further show that this equivalence holds for both cascaded and parallel tokenizations, the two most widely used schemes in industrial GR systems. Our result provides the first formal theoretical foundation for the dominant industrial GR paradigm, and offers principled guidance for future GR system optimization.

\end{abstract}

\maketitle

\section{Introduction}
Generative recommendation (GR) has emerged as a transformer-based architecture in sequential recommendation over the past few years, attracting great attention from both academia and industry. Recently, several industrial scale GR systems (e.g. OneRec\cite{deng2025onerec}, OneLive\cite{wang2026onelive}, OneSug\cite{guo2026onesug}, OneMall\cite{zhang2026onemall}, GPR\cite{zhang2025gpr} and DOS\cite{yin2026dual}, etc.) have been deployed with promising business metrics gain, showing the effectiveness of the GR paradigm.

Following the pioneering work of Tiger \cite{rajput2023recommender}, the mainstream GR paradigm has converged to a standardized framework\cite{hou2025survey}: it adopts semantic-ID (SID) tokenization for candidate item indexing, uses next-token prediction (NTP) as the core training objective for next-item modeling, and generates the target candidate via auto-regressive (AR) decoding at inference time. This three-stage pipeline has become a de facto standard for GR system design, as validated by its universal adoption in the aforementioned industrial deployments.

However, existing GR works mainly focus on model architecture design and empirical performance optimization, with few rigorous theoretical explanations for the working mechanism of NTP in GR. This gap limits both the academic understanding of the GR paradigm and the industrial optimization of GR systems.

The most relevant theoretical progress comes from the natural language processing domain: recent work \cite{malach2024auto} formally proves that AR next-token predictors are universal function learners, and attributes the strong generalization capability of modern large language models (LLMs) primarily to the AR-NTP training scheme. In the GR domain, concurrent work \cite{ding2026well} systematically verifies that GR’s superiority over conventional ID-based models stems from better generalization, and links GR’s item-level generalization to token-level memorization via controlled experiments; paper\cite{zhao2026farewell} further shows that semantic tokens have greater scaling potential compared to item IDs in large-scale GR systems.

Despite these advances, there is still a lack of rigorous mathematical explanation for the working mechanism of AR-NTP in GR from the perspective of recommendation systems. Hence, in this paper, we provide this short note to formally prove \textbf{the equivalence between auto-regressive next-token prediction (AR-NTP) and full-item-vocabulary maximum likelihood estimation (FV-MLE)} for generative recommendation.

Our contributions are listed below:
\begin{itemize}[topsep=2pt, itemsep=1pt, leftmargin=*]
    \item We provide the first formal theoretical foundation for the standard GR paradigm, by rigorously proving that the widely used k-token AR-NTP is strictly mathematically equivalent to FV-MLE, under the core premise of a bijective mapping between items and their corresponding k-token sequences. This resolves the long-standing gap between GR’s widespread empirical success and the lack of formal theoretical analysis.
    \item We further prove that this equivalence holds for both cascaded and parallel tokenization, the two mainstream tokenization schemes adopted in industrial GR systems.
\end{itemize}

\section{Notation and Definitions}

In this section, we define notations and introduce concepts in generative recommendation system. Main notations are listed below:
\begin{itemize}[topsep=2pt, itemsep=1pt, leftmargin=*]
    \item $h \in \mathbb{R}^d$: User and context representation encoded from the user's historical behavior sequence, user-side static features and request's context features.
    \item $\mathcal{V}$: Full item vocabulary, with total size of items $N = |\mathcal{V}|$.
    \item $k$: Number of tokens used to represent a single candidate item. $k > 0$.
    \item $X$: Size of each token's codebook. For simplicity, we assume that $N = X^k$ which ensures full coverage of the item vocabulary and we use the same codebook size for all tokens.
    \item $\mathcal{V}_1, \mathcal{V}_2, \dots, \mathcal{V}_k$: Token codebooks for each token position, with $|\mathcal{V}_m| = X$ for all $m \in [1,k]$.
    % \item $\Phi(i) = (t_1^i, t_2^i, \dots, t_k^i)$ $t_i\leftrightarrow[t_1^i, t_2^i, \dots, t_k^i]$: The k-token sequence for item $i$, with an \textbf{one-to-one mapping} between items and valid k-token sequences.
    \item $\Phi(i) = (t_1^i, t_2^i, \dots, t_k^i)$: The k-token sequence for item $i$, with an \textbf{bijection mapping}, i.e. $\Phi: \mathcal{V} \leftrightarrow \mathcal{S}$ between item space $V$ and valid k-token sequences space $S$.
    \item $\ell(t_m \mid h, t_1, \dots, t_{m-1})$: Conditional predict logit of the $m$-th token given user context $h$ and all previous tokens, output by the generative recommendation model through auto-regressive next token prediction (NTP).
    \item $\ell(h,i)$: Item $i$'s logit. In NTP setting, we define it as the sum of token-level conditional logits:
    $$\ell(h,i) = \sum_{m=1}^k \ell(t_m^i \mid h, t_1^i, \dots, t_{m-1}^i)$$
\end{itemize}

In this paper, we focus on \textbf{single-next-item generation}: given a user's historical sequence and the request context features which is encoded into embedding $h$, the model generates the $k$-token sequence corresponding to the target item via step-by-step auto-regressive prediction (list-wise generation is out of the scope of this paper).
\begin{itemize}[leftmargin=*, noitemsep, topsep=0pt]
    \item \textbf{Training Pipeline}: In the NTP paradigm, each training sample is treated as positive sample $i^+$ (e.g., the next exposed item for exposure modeling or the next interacted item for user preference modeling). The model is trained to predict the token sequence of this positive item in an auto-regressive manner: each step predicts the next token conditioned on the user context $h$ and all previously generated tokens of the sequence.
    \item \textbf{Inference Pipeline}: Given $h$, the recommendation model generates the most likely $k$-token sequences via beam search, which are then mapped back to the corresponding items by token-item mapping as the final results.
\end{itemize}

We formalize this single-next-item $k$-token auto-regressive NTP framework with the following definition:

\begin{definition}[\textbf{Single-Next-Item $k$-Token Auto-Regressive NTP}]
$k$-token NTP models the joint generative probability of the target item given user\&context embedding $h$ via Bayesian chain rule. The item-level generative probability could be decomposed into a product of token-level conditional probabilities:
\begin{equation}
p(i \mid h) = p(t_1^i, t_2^i, \dots, t_k^i \mid h) = \prod_{m=1}^k p(t_m^i \mid h, t_1^i, \dots, t_{m-1}^i)
\end{equation}
Each token-level conditional probability is parameterized by a full-codebook Softmax:
\begin{equation}
p(t_m \mid h, t_1, \dots, t_{m-1}) = \frac{\exp\left(\ell(t_m \mid h, t_1, \dots, t_{m-1})\right)}{Z_m(h, t_1, \dots, t_{m-1})}
\end{equation}
where $Z_m(\cdot)$ is the token-level partition function for the $m$-th codebook:
\begin{equation}
Z_m(h, t_1, \dots, t_{m-1}) = \sum_{t_m \in \mathcal{V}_m} \exp\left(\ell(t_m \mid h, t_1, \dots, t_{m-1})\right).
\end{equation}
For model training, standard \textbf{teacher forcing} strategy is usually adopted for auto-regressive learning: when predicting the $m$-th token of the target item sequence, the model takes the ground truth preceding tokens of the target item $[t_1^{i^+}, \dots, t_{m-1}^{i^+}]$ as conditional input, rather than tokens generated by the model in previous steps. Under this training strategy, the final negative log-likelihood (NLL) loss of $k$-token NTP is the sum of per-step token prediction losses on the target item's ($i^+$) ground truth token sequence $[t^{i^+}_1,\cdots,t^{i^+}_k]$:
\begin{equation}
\mathcal{L}_{\text{NTP}} = -\sum_{m=1}^k \log p(t_m^{i^+} \mid h, t_1^{i^+}, \dots, t_{m-1}^{i^+}).
\end{equation}

\end{definition}

Having formalized the single-next-item $k$-token auto-regressive NTP framework, we center this paper on a core theoretical question: \textbf{what is the fundamental nature of this widely adopted NTP paradigm}?

To rigorously answer this question, we first establish a theoretically grounded benchmark for distribution fitting over the full item vocabulary: full-item-vocabulary Maximum Likelihood Estimation (MLE). In statistics, MLE is a standard, well-justified criterion for fitting a model distribution to observed data, with rigorously proven asymptotic properties: as the sample size increases, the MLE estimator converges in probability to the true underlying data-generating distribution, making it the de facto gold standard for evaluating distribution fitting performance.

We formalize this benchmark objective for recommendation retrieval tasks in the following definition.

\begin{definition}[\textbf{Full-Item-Vocabulary MLE}]
Full-item-vocabulary maximum likelihood estimation (MLE) is an objective for fitting the target distribution over the full item vocabulary given a user and context embedding $h$. It maximizes the log-likelihood of positive samples from the target distribution, with the likelihood defined over the full item vocabulary. The corresponding Negative Log-Likelihood (NLL) loss is:
\begin{equation}
\mathcal{L}_{\text{Full\_MLE}} = -\log\left( \frac{\exp\left(\ell(h, i^+)\right)}{Z_{\text{full}}(h)} \right)
\label{loss_MLE}
\end{equation}
where:
\begin{itemize}[leftmargin=*, noitemsep]
    \item $i^+$ denotes a positive sample from the target distribution, where the definition of it is fully determined by the objective of the recommendation task (e.g., exposed item, user-clicked item, or other target behavior samples aligned with the task goal);
    \item $\ell(h, i)$ is the logits of item $i$;
    \item $Z_{\text{full}}(h)$ is the full-item-vocabulary partition function, which normalizes logits into a valid probability distribution over the entire item vocabulary:
    \begin{equation}
    Z_{\text{full}}(h) = \sum_{i \in \mathcal{V}} \exp\left(\ell(h, i)\right)
    \end{equation}
\end{itemize}
\end{definition}

\section{Proof of the Main Equivalence Theorem}
In this section, we formally prove our core theoretical result: the single-next-item $k$-token auto-regressive NTP paradigm, widely adopted in GR, is strictly mathematically equivalent to full-item-vocabulary MLE. This result provides a rigorous theoretical foundation for the empirical effectiveness of NTP-based GR model.

\begin{theorem}[Equivalence of $k$-Token NTP and Full-Vocabulary MLE]
Assume that there is a \textbf{bijective mapping} between the full item vocabulary $\mathcal{V}$ and the Cartesian product of the token codebooks $\mathcal{V}_1 \times \mathcal{V}_2 \times \dots \times \mathcal{V}_k$, i.e., each item $i \in \mathcal{V}$ uniquely maps to a $k$-token sequence $(t_1^i, t_2^i, \dots, t_k^i)$, and vice versa. Then the negative log-likelihood (NLL) loss of $k$-token auto-regressive NTP is \textbf{strictly equivalent} to the NLL loss of full-item-vocabulary MLE:
$$\mathcal{L}_{\text{NTP}} \equiv \mathcal{L}_{\text{Full\_MLE}}$$
\end{theorem}

\begin{proof}
We provide the proof in the following four steps.

\paragraph{Step 1: Decompose the item-level joint probability of NTP}
Using the Bayesian chain rule for auto-regressive generation, the conditional probability of item $i$ given user context $h$ is decomposed as:
\begin{equation}
p(i^+ \mid h) = p(t_1^{i^+}, t_2^{i^+}, \dots, t_k^{i^+} \mid h) = \prod_{m=1}^k p(t_m^{i^+} \mid h, t_1^{i^+}, \dots, t_{m-1}^{i^+})
\end{equation}
Substitute the token-level full-codebook softmax definition into the product:
\begin{equation}
p(i^+ \mid h) = \prod_{m=1}^k \frac{\exp\left(\ell(t_m^{i^+} \mid h, t_1^{i^+}, \dots, t_{m-1}^{i^+})\right)}{Z_m(h, t_1^{i^+}, \dots, t_{m-1}^{i^+})}
\end{equation}
where $\ell(\cdot)$ denotes the token-level conditional logit, and $Z_m(\cdot) = \sum_{t_m \in \mathcal{V}_m} \exp\left(\ell(t_m \mid h, t_1, \dots, t_{m-1})\right)$ is the token-level partition function for the $m$-th decoding step.

\paragraph{Step 2: Merge token-level logits into the item-level joint logit}
By the exponential product rule, the numerator of the joint probability is merged into a single exponential term, consistent with our item-level joint logit definition:
\begin{align}
\prod_{m=1}^k \exp\left(\ell(t_m^{i^+} \mid h, t_1^{i^+}, \dots, t_{m-1}^{i^+})\right) & = \exp\left( \sum_{m=1}^k \ell(t_m^{i^+} \mid h, t_1^{i^+}, \dots, t_{m-1}^{i^+}) \right)\\ & = \exp\left(\ell(h,i^+)\right)
\end{align}
We thus rewrite the item-level conditional probability as:
\begin{equation}
p(i^+ \mid h) = \frac{\exp\left(\ell(h,i)\right)}{\prod_{m=1}^k Z_m(h, t_1^{i^+}, \dots, t_{m-1}^{i^+})}
\label{eq.NTP_ZM}
\end{equation}

\paragraph{Step 3: Prove the equivalence of partition functions}
We expand the product of token-level partition functions via the distributive law of multiplication over addition:
% \begin{equation}
% \prod_{m=1}^k Z_m(h, t_1, \dots, t_{m-1}) = \sum_{t_1 \in \mathcal{V}_1} \exp\left(\ell(t_1 \mid h)\right) \cdot \sum_{t_2 \in \mathcal{V}_2} \exp\left(\ell(t_2 \mid h, t_1)\right) \cdot \dots \cdot \sum_{t_k \in \mathcal{V}_k} \exp\left(\ell(t_k \mid h, t_1, \dots, t_{k-1})\right)
% \end{equation}
\begin{equation}
\prod_{m=1}^k Z_m(h, t_1, \dots, t_{m-1}) = \sum_{t_1 \in \mathcal{V}_1} \exp\left(\ell(t_1 \mid h)\right) \cdot \sum_{t_2 \in \mathcal{V}_2} \exp\left(\ell(t_2 \mid h, t_1)\right) \cdots
\end{equation}
Merging the nested sums into a single summation over all valid $k$-token sequences:
\begin{equation}
\prod_{m=1}^k Z_m(\cdot) = \sum_{t_1 \in \mathcal{V}_1} \sum_{t_2 \in \mathcal{V}_2} \dots \sum_{t_k \in \mathcal{V}_k} \exp\left( \sum_{m=1}^k \ell(t_m \mid h, t_1, \dots, t_{m-1}) \right)
\end{equation}
By the bijective mapping premise in Theorem 1, summing over all valid $k$-token sequences is exactly equivalent to summing over all items in the full vocabulary $\mathcal{V}$. Substituting the item-level joint logit definition, we get:
\begin{equation}
\prod_{m=1}^k Z_m(\cdot) = \sum_{i \in \mathcal{V}} \exp\left(\ell(h,i)\right) = Z_{\text{full}}(h)
\end{equation}
% This strictly proves that the product of token-level partition functions is a global constant for fixed $h$, independent of the specific item $i$, and exactly equal to the full-item-vocabulary partition function $Z_{\text{full}}(h)$.

This strictly proves that the product of token-level partition functions is exactly equal to the full-item-vocabulary partition function $Z_{\text{full}}(h)$.

\paragraph{Step 4: Prove the equivalence of loss functions}
Substitute the partition function equivalence result back into the item-level conditional probability of NTP in Eq.~\ref{eq.NTP_ZM}:
\begin{equation}
p(i^+ \mid h) = \frac{\exp\left(\ell(h,i)\right)}{Z_{\text{full}}(h)}
\end{equation}
This is exactly the conditional probability definition of full-item-vocabulary MLE. For the training sample $i^+$, the NLL loss of NTP is:
\begin{equation}
\mathcal{L}_{\text{NTP}} = -\log p(i^+ \mid h) = -\log\left( \frac{\exp\left(\ell(h,i^+)\right)}{Z_{\text{full}}(h)} \right) = \mathcal{L}_{\text{MLE}}
\end{equation}
This completes the full proof of the equivalence theorem.
\end{proof}

\paragraph{Computational Complexity Advantage}
Our equivalence theorem reveals the core reason for the widespread industrial adoption of the AR-NTP-based GR paradigm: it achieves the theoretical optimality of FV-MLE, while avoiding the prohibitive computational complexity of full-item softmax:
\begin{itemize}[topsep=2pt, itemsep=1pt, leftmargin=*]
\item The standard FV-MLE requires a full softmax normalization with computational complexity $O(V)$, which is infeasible for industrial-scale $V$ (typically $10^6$ to $10^9$).
\item In contrast, the k-token NTP paradigm decomposes the full-item softmax into k sequential token-level softmax operations, each with complexity $O(V_m)$ (where $V_m$ is the size of each token codebook, typically $256\sim8192$ in industrial practice). The overall computational complexity is reduced to $O(k \cdot V_m)$, which is much lower than FV-MLE for industrial-scale item space.
\end{itemize}

Our theorem proves that this complexity reduction comes with no loss of theoretical optimality: as long as the bijective mapping condition is satisfied, AR-NTP is strictly equivalent to FV-MLE.

\section{Corollary: Applicability to Cascaded and Parallel Tokenization}
In GR, two tokenization paradigms dominate mainstream industrial system design:
\begin{itemize}[topsep=2pt, itemsep=1pt, leftmargin=*]
    \item Cascaded tokenization, implemented via residual quantization-based methods such as RQ-VAE\cite{rajput2023recommender} and RQ-KMeans\cite{deng2025onerec}, which has long been a standard practice in GR systems.
    \item Parallel tokenization\cite{hou2025generating} based on product quantization (e.g., optimized product quantization, OPQ), an emerging paradigm in recent advances that achieves promising performance while enabling efficient inference via multi-token prediction (MTP).
\end{itemize}
Notably, while parallel tokenizers are compatible with MTP for fast decoding, they also fully support AR decoding. It is therefore critical to establish that our core equivalence theorem holds for both of these mainstream tokenization schemes, which extends the generalizability of our theoretical result. Hence, we derive the following corollary to formalize the universal applicability of our result:

\begin{corollary}[Equivalence Applicability to Parallel \& Cascaded tokenization Schemes]
% Let the full item vocabulary $\mathcal{V}$ and fixed token sequence length $k$ be consistent with the main theorem. 
For both:
\begin{itemize}[topsep=2pt, itemsep=1pt, leftmargin=*]
    \item \textbf{Parallel tokenization}: where the $m$-th token codebook $\mathcal{V}_m$ is independent of tokens from other codebooks, i.e., $t_m \perp\!\!\!\perp t_j, \forall j \neq m$;
    \item \textbf{Cascaded tokenization}: where the $m$-th token codebook $\mathcal{V}_m$ is dependent on all preceding tokens($t_1, \dots, t_{m-1}$), i.e., $t_m \mid t_1, \dots, t_{m-1}$;
\end{itemize}
the mathematical equivalence between $k$-token AR-NTP and FV-MLE holds, as long as there exists a bijective mapping $\Phi: \mathcal{V} \to \mathcal{S}$, where $\mathcal{S}$ is the valid $k$-token sequence space generated by the corresponding tokenization scheme.
\end{corollary}

\begin{proof}
We first restate the two sufficient and necessary conditions for the core equivalence theorem to hold, which could be derived from the complete proof of the main theorem:
\begin{enumerate}[label=(\roman*), leftmargin=*, noitemsep, topsep=2pt, itemsep=1pt]
    \item \textbf{Bijective Mapping Condition}: There exists a bijection between the full item vocabulary $\mathcal{V}$ and the valid $k$-token sequence space $\mathcal{S}$, denoted as:
    $$\Phi: \mathcal{V} \leftrightarrow \mathcal{S}$$
    where each item $i \in \mathcal{V}$ maps to a unique $k$-token sequence $\Phi(i) = (t_1^i, t_2^i, \dots, t_k^i) \in \mathcal{S}$, and vice versa.
    \item \textbf{Probability Decomposition Condition}: The item-level conditional probability is modeled via the Bayesian chain rule for auto-regressive generation:
    \begin{equation}
    p(i \mid h) = \prod_{m=1}^k p(t_m^i \mid h, t_1^i, \dots, t_{m-1}^i)
    \end{equation}
\end{enumerate}

In industrial practice, the satisfaction of the bijective mapping condition depends on the quality of the quantization algorithm: suboptimal GR designs which lead to codebook collapse or high collision rate will violate the bijection condition, while some methods like RecGPT\cite{jiang2025recgpt} apply finite scalar quantization (FSQ) to guarantee collision-free one-to-one mapping. For the brevity of the following proof, we assume the bijective mapping condition holds for all discussed tokenization schemes.

We now prove that: both cascaded and parallel tokenizations satisfy the above two conditions in the AR-NTP setting, hence the equivalence theorem applies.

\paragraph{Proof for Cascaded Tokenization}
By definition, cascaded tokenization (e.g., RQ-VAE, RQ-KMeans) generates the $m$-th token conditioned on all preceding tokens, which exactly matches the conditional probability form in Condition (ii). Under our bijective mapping assumption, Condition (i) is also satisfied. The item-level joint probability decomposition of cascaded coding is identical to the auto-regressive form used in the main theorem, so all derivation steps of the main theorem hold without modification. Thus, the equivalence between AR-NTP and full-vocabulary MLE holds for cascaded tokenization.

\paragraph{Proof for Parallel Tokenization}
Under our bijective mapping assumption, Condition (i) is strictly satisfied for valid parallel tokenization designs. For Condition (ii), parallel tokenization is a special case of the auto-regressive decomposition, where the token-level conditional probability is independent of preceding tokens, i.e.:
\begin{equation}
p(t_m^i \mid h, t_1^i, \dots, t_{m-1}^i) = p(t_m^i \mid h), \quad \forall m \in \{1,2,\dots,k\}
\end{equation}
This form fully conforms to the Bayesian chain rule in Condition (ii), and we derive the item-level probability of parallel tokenization via the following steps:
\begin{enumerate}[leftmargin=*, noitemsep, topsep=2pt, itemsep=1pt]
\item Substituting the conditional independence form into the item-level joint probability gives:
\begin{equation}
p(i \mid h) = \prod_{m=1}^k p(t_m^i \mid h) = \prod_{m=1}^k \frac{\exp\left(\ell(t_m^i \mid h)\right)}{Z_m(h)}
\end{equation}
where $Z_m(h) = \sum_{t_m \in \mathcal{V}_m} \exp\left(\ell(t_m \mid h)\right)$ is the token-level partition function for the $m$-th independent codebook.
\item Merge the numerator into the item-level joint logit:
$$\prod_{m=1}^k \exp\left(\ell(t_m^i \mid h)\right) = \exp\left(\sum_{m=1}^k \ell(t_m^i \mid h)\right) = \exp\left(\ell(h,i)\right)$$
\item By the distributive law of multiplication over addition, expand the product of token-level partition functions as:
    \begin{align}
    \prod_{m=1}^k Z_m(h) & = \prod_{m=1}^k \sum_{t_m \in \mathcal{V}_m} \exp\left(\ell(t_m \mid h)\right)\\
     & = \sum_{t_1 \in \mathcal{V}_1} \dots \sum_{t_k \in \mathcal{V}_k} \exp\left(\sum_{m=1}^k \ell(t_m \mid h)\right)
    \end{align}
\item By the bijective mapping condition, summing over all valid $k$-token sequences is strictly equivalent to summing over the full item vocabulary $\mathcal{V}$, hence:
    \begin{equation}
    \prod_{m=1}^k Z_m(h) = \sum_{i \in \mathcal{V}} \exp\left(\ell(h,i)\right) = Z_{\text{full}}(h)
    \end{equation}
\end{enumerate}
The final item-level probability is identical to the full-vocabulary MLE (FV-MLE) form in Eq.~\ref{loss_MLE}:
\begin{equation}
p(i \mid h) = \frac{\exp\left(\ell(h,i)\right)}{Z_{\text{full}}(h)}
\end{equation}
Thus the strict mathematical equivalence between AR-NTP and FV-MLE is satisfied for parallel tokenization.
\end{proof}

This completes the full proof of the corollary.

\section{Conclusion}
In this paper, we address the long-standing open question of why the AR-NTP paradigm delivers consistent empirical effectiveness in GR by demonstrating a mathematical equivalence between the widely adopted k-token AR-NTP paradigm and FV-MLE. We further extend this core result to the two most dominant tokenization schemes in industrial GR systems: cascaded tokenization and parallel tokenization, verifying the universal applicability of our equivalence theorem. Our equivalence theorem reveals the essential nature of GR's generalization advantage: \textbf{the AR-NTP paradigm inherently achieves unbiased fitting of the full item distribution}. 

For future work, we plan to develop a theoretical analysis of the performance boundary of GR when the bijective mapping condition is violated (e.g., codebook collapse, token sequence collision), to quantify the impact of tokenization quality on GR performance.

\bibliographystyle{ACM-Reference-Format}
\bibliography{ref}

\end{document}